\documentclass{amsart}
 \usepackage{amsaddr}
\usepackage[english]{babel}
\usepackage{comment}
\usepackage[letterpaper,top=2cm,bottom=2cm,left=3cm,right=3cm,marginparwidth=1.75cm]{geometry}
\usepackage{todonotes}

\usepackage{amsmath}
\usepackage{graphicx}
\usepackage[colorlinks=true, allcolors=blue]{hyperref}
\usepackage{cleveref}

\DeclareMathOperator{\indeg}{indeg}
\DeclareMathOperator{\outdeg}{outdeg}

\usepackage{tikz}
\usepackage{tkz-graph}
\usetikzlibrary{shapes,positioning,calc}
\tikzstyle{tre}=[circle,draw,inner sep = 0pt, minimum size=1.5mm]
\tikzstyle{pe}=[circle,draw,inner sep=0pt,minimum size=1.5mm]
\tikzstyle{mi}=[circle,draw,inner sep=0pt,minimum size=1.5mm, fill = black!60]
\tikzstyle{mi2}=[circle,color = lightgray,draw,inner sep=0pt,minimum size=1.5mm, fill = lightgray]
\tikzstyle{mini}=[circle,draw,inner sep=0pt,minimum size=0.8mm, fill = black!60]
\tikzstyle{minil}=[circle,draw,color = lightgray,inner sep=0pt,minimum size=0.8mm, fill = lightgray]
\tikzstyle{minir}=[circle,color = red,draw,inner sep=0pt,minimum size=0.8mm, fill = red!60]
\tikzstyle{minib}=[circle,color = blue,draw,inner sep=0pt,minimum size=0.8mm, fill = blue!60]
\tikzstyle{mimarc}=[circle,draw,inner sep=0pt,minimum size=1.5mm, fill = red!60]
\tikzstyle{bxmarc}=[rectangle,draw,inner sep=0pt,line width=0.2pt,minimum size=1.5mm, fill = red!60]
\tikzstyle{bb}=[circle,draw,inner sep=0pt,line width=1.5pt,minimum size=4.5mm]
\tikzstyle{rbx}=[circle,fill,draw,inner sep=0pt,minimum size=1.5mm]
\tikzstyle{bx}=[rectangle,draw,inner sep=0pt,line width=0.2pt,minimum size=1.5mm, fill = black!60]
\tikzstyle{bxb}=[rectangle,draw,inner sep=0pt,line width=0.2pt,minimum size=1.5mm, fill = blue!70]
\tikzstyle{minibx}=[rectangle,draw,inner sep=0pt,line width=0.2pt,minimum size=0.8mm, fill = black!60]
\tikzstyle{minilbx}=[rectangle,color = lightgray,draw,inner sep=0pt,line width=0.2pt,minimum size=0.8mm, fill = lightgray]
\tikzstyle{minirbx}=[rectangle,color = red,draw,inner sep=0pt,line width=0.2pt,minimum size=0.8mm, fill = red!60]
\tikzstyle{minibbx}=[rectangle,color = blue,draw,inner sep=0pt,line width=0.2pt,minimum size=0.8mm, fill = blue!60]
\tikzstyle{label}=[circle,inner sep=0pt,minimum size=0.2mm, fill = white]
\tikzstyle{ppalpath}=[style={->},
decoration={snake,pre length=0.05cm, post length=0.05cm,
amplitude=.4mm,segment length=2mm},
decorate]
\tikzset{
    between/.style args={#1 and #2}{
         at = ($(#1)!0.5!(#2)$)
    }
}
\tikzset{
    amunt/.style args={#1 and #2}{
         at = ($(#1)!0.25!(#2)$)
    }
}

\newtheorem{thm}{Theorem}
\newtheorem{prop}[thm]{Proposition}
\newtheorem{lem}[thm]{Lemma}
\newtheorem{cor}[thm]{Corollary}
\newtheorem{conj}[thm]{Conjecture}
\newtheorem{defn}[thm]{Definition}

\newenvironment{pf}{\begin{proof}}{\end{proof}}

\DeclareMathOperator{\head}{head}
\DeclareMathOperator{\tail}{tail}

\title{Fence Decompositions and Cherry Covers in Non-Binary Phylogenetic Networks}

\author{%
    Joan Carles Pons\textsuperscript{a,*}, Pau Vives López\textsuperscript{a}, Yukihiro Murakami\textsuperscript{b}, Leo van Iersel\textsuperscript{b}.
     \\[1em]
    \textsuperscript{a} Department of Mathematics and Computer Science, \\ 
    Universitat de les Illes Balears, Spain \\[1em]
        \textsuperscript{b} Delft Institute of Applied Mathematics, Delft University of Technology \\[1em]
    \textsuperscript{*} Corresponding author: joancarles.pons@uib.es
}

\begin{document}

\begin{abstract}
Reticulate evolution can be 
modelled using phylogenetic networks. 
Tree-based networks, which are one of the more general classes of phylogenetic networks, 
have recently gained eminence for its ability to represent evolutionary histories with an underlying tree structure. 
To better understand tree-based networks, numerous characterizations have been proposed, based on tree embeddings, matchings, and arc partitions.
Here, we build a bridge between two arc partition characterizations, namely maximal fence decompositions and cherry covers.
Results on cherry covers have been found for general phylogenetic networks.
We first show that the number of cherry covers is the same as the number of support trees (underlying tree structure of tree-based networks) for a given semi-binary network. 
Maximal fence decompositions have only been defined thus far for binary networks (constraints on vertex degrees). 
We remedy this by generalizing fence decompositions to 
non-binary networks, and using this, we characterize semi-binary tree-based networks in terms of forbidden structures.
Furthermore, we give an explicit enumeration of cherry covers of semi-binary networks, by studying its fence decomposition. 
Finally, we prove that it is possible to characterize semi-binary tree-child networks, a subclass of tree-based networks, in terms of the number of their cherry covers.
\end{abstract}

\maketitle

\section{Introduction}

Phylogenetic trees serve as a tool for depicting non-reticulate evolution (e.g.\cite{nei2000molecular}). Their inherent simplicity contrasts with the impossibility of representing more complex evolutionary scenarios involving reticulate events like hybridizations, recombinations, or lateral gene transfers. Phylogenetic networks extend phylogenetic trees, allowing for the representation of these more complex situations \cite{huson2010phylogenetic,bapteste2013networks}.

The space of phylogenetic networks is divided into classes to reproduce different biologically significant scenarios or to explore challenging mathematical and computational problems. 
One can find a review of such classes in \cite{kong2022classes}. 
One of the ongoing discourses in literature is the debate regarding the nature of evolution. 
It revolves around whether evolution predominantly follows a tree-like pattern, punctuated by occasional horizontal occurrences \cite{daubin2003phylogenetics}, or if evolution inherently exhibits network-like characteristics without any resemblance to a tree structure \cite{dagan2006tree}. 
Tree-based networks \cite{francis2015phylogenetic} were introduced to represent the networks of the first type.  
In graph theory terms, such networks are those that have a rooted spanning tree (called a \emph{support tree}) for which the leaves are those of the network. Suppressing elementary nodes (nodes of indegree and outdegree equal to one) in the support tree results in a \emph{base-tree} as the underlying tree structure for the network. 
In general, a tree-based network can have more than one support tree, and possibly more than one base tree (different support trees can give rise to the same base trees). 
Tree-based networks include several relevant subclasses of phylogenetic networks, including, among others, tree-child networks \cite{cardona2008comparison}. 
Equivalently, a network is tree-child if and only if each embedded tree, that is,  is a support tree.

In addition to the network classes mentioned above, one typically characterizes phylogenetic networks based on the degree of vertices, that is, the number of arcs that every vertex is allowed to be incident to.
This characterization is useful, and oftentimes necessary when there are ambiguities in the order of how evolutionary events have unfolded.
These ambiguities are caused by insufficient data (\emph{soft polytomies}) or a simultaneous divergence of multiple species from a single speciation event (\emph{hard polytomies})~\cite{sayyari2018testing}; we also encounter the same ambiguities in reticulate evolution~\cite{ottenburghs2019multispecies}.
Roughly speaking, one considers \emph{binary} networks when assuming no ambiguities are present, \emph{semi-binary} networks when assuming ambiguities are present for reticulate evolution, and \emph{non-binary} networks when assuming ambiguities are present in both speciations and reticulate evolution.
Here, non-binary means that the network is `not necessarily binary', and thus the class of binary networks are contained in the class of semi-binary networks, which is itself contained in the class of non-binary networks.
A formal definition for the three classes is given in \Cref{subsec:PhyloNet}.
Note that in the literature, \emph{semi-binary} is sometimes used to describe networks that have ambiguities in speciation events~\cite{jetten2016nonbinary}. 



In this paper, we study and build a bridge between two characterizations of tree-based networks called \emph{cherry covers} and \emph{fence decompositions}, which are both based on arc partitions.  
The concept of cherry covers relative to phylogenetic networks was introduced in  \cite{van2021unifying}. 
Roughly speaking, a cherry cover refers to a partition of the network arcs into sets of sizes two and three, such that each of these sets contains arcs that share endpoints in some way.
It was shown that non-binary networks are tree-based if and only if they have a cherry cover.
On the other hand, fence decompositions of phylogenetic networks were introduced in \cite{zhang2016tree,hayamizu2021structure}, and similarly to cherry covers, it outlines a systematic approach to breaking down any network into its ``fundamental'' substructures.
Roughly speaking, it is an arc partition in which every set is an ``up-down path'' of arcs, meaning that each arc shares at least one and at most two endpoints with other arcs in the set.
This partition can be used to characterize binary tree-based networks in terms of forbidden structures, or to count support trees in binary tree-based networks. 
We give formal definitions and known results on cherry covers and fence decompositions in \Cref{subs:cherry} and \Cref{subsec:bindecomp}, respectively.




Both cherry covers and fence decomposition have relevant implications for the study of tree-based networks; even so, the concepts have not been studied in a framework in which both can be handled together. 
In \Cref{sec:cv_suppT}, we prove that the number of cherry covers is the same as the number of support trees for semi-binary tree-based networks. 
In Section \ref{sec:fences_vs_cv}, we show that the number of cherry covers (and thus also the number of support trees) of a binary tree-based network can be written as a formula involving the elements of its fence decomposition.
Notice that the problem of counting the number of support trees in binary tree-based networks have been previously solved in \cite{jetten2016nonbinary, pons2019tree} by the association of a bipartite graph to the network and in \cite{hayamizu2021structure} by the use of the fence decomposition.  
In Section \ref{sec:extension} we extend the fence decomposition of binary phylogenetic networks introduced in \cite{hayamizu2021structure} to non-binary networks. This requires the introduction of a new set of fence structures, and in so doing, we extend problems defined for binary networks to the non-binary setting.
In Section \ref{subsec:CharacterizationSBTBN} we solve two of these problems. 
We extend to semi-binary tree-based networks (a) the problem of counting the number of their cherry covers (and therefore the number of their support trees) from the ``new'' fence decomposition, mimicking the results from Section \ref{sec:fences_vs_cv}, and (b) the characterisation of semi-binary tree-based networks in terms of forbidden structures. 
Both generalisations require similar approaches, albeit much more technical and tricky than their binary counterparts. 
Finally, in Section \ref{sec:TC}, we prove that the subclass of semi-binary tree-child networks can be characterised in terms of the number of their cherry covers. The paper ends with some concluding remarks in \Cref{sec:Conc}.

\section{Preliminaries}\label{sec:Prelim}

In this section we give some definitions and results that will be used throughout the manuscript. 

\subsection{Phylogenetic Networks}\label{subsec:PhyloNet}

    A \emph{(non-binary) phylogenetic network}, or simply a \emph{network}, $N=(V,A)$
    on a set $X$ of \emph{taxa}, is a directed acyclic graph $(V,A)$ without parallel arcs such that any node $u\in V$ is either: 
    \begin{enumerate}
        \item a \emph{root}, with $\indeg u=0$, $\outdeg u=1$ (and there can only be one such node), or
        \item a \emph{leaf}, with $\indeg u=1$, $\outdeg u=0$, or
        \item a \emph{tree node}, with $\indeg u=1$, $\outdeg u\geq 2$, or
        \item a \emph{reticulation}, with $\indeg u\geq 2$, $\outdeg u=1$,
    \end{enumerate}
    together with a  bijective function between $X$ and the set of leaves. Vertices in $V\setminus X$ are called \emph{interior nodes}.

The arc incident to the root is called the \emph{root arc}, and arcs that feed into a reticulation node are called \emph{reticulation arcs}.
A network is \emph{semi-binary} if all tree nodes have outdegree $2$ (reticulations with indegree greater than $2$ are allowed), and it is \emph{binary} if, in addition to these conditions, all reticulations have indegree $2$. 
We write non-binary to mean \emph{not necessarily binary}. So the class of binary networks are contained in the class of semi-binary networks, which are themselves contained in the class of non-binary networks.

In this paper we focus on some relevant subclasses of phylogenetic networks, mainly tree-based and tree-child networks. 


A (non-binary) network $N$ is \emph{tree-based} \cite{francis2015phylogenetic,jetten2016nonbinary} with \emph{base-tree} $T$ when $N$ can be obtained from $T$ by the following process: 
\begin{enumerate}
    \item subdivide some arcs of $T$ generating elementary (indegree and outdegree equal to one) nodes called \emph{attachment points};
    \item add new arcs between pairs of attachment points and from nodes in $T$ to attachment nodes, in such a way that $N$ remains acyclic, without parallel arcs, and such that each attachment point maintains indegree or outdegree 1;
    \item suppress every elementary node. 
\end{enumerate}
Removing all arcs added in step (2) gives a subdivision of~$T$, called a \emph{support tree} of $N$.
In other terms, a tree (possibly containing elementary nodes) obtained from $N$ by removing all but one incoming reticulation arcs from each reticulation without creating new leaves is a support tree. 
Note that the set of vertices of a support tree for $N$ is also the set of vertices of $N$.


A (non-binary) network is \emph{tree-child} \cite{cardona2008comparison} if every interior node has at least one child which is a tree node or a leaf. The class of tree-child networks is included in the class of tree-based networks. 


\subsection{Cherry covers} 
\label{subs:cherry}

We adapt here some definitions and results from \cite{van2021unifying}. Let $N=(V,A)$ be a semi-binary phylogenetic network.



A \emph{cherry shape} in $N$ is a subgraph of $N$ composed by three different nodes $x,y,p \in V$ and two arcs $px, py \in A$. We denote the cherry shape by its set of arcs $\{px, py\}$. We will refer to $p$ as the \emph{internal} node and to $x$ and $y$ as the \emph{terminal} nodes.
A \emph{reticulated cherry shape} in $N$ is a subgraph of $N$ composed by four different nodes $x,y,p_x,p_y \in V$ and three arcs $p_xx,p_yp_x,p_yy\in A$, with $p_y$ a tree node and $p_x$ a reticulation. Similarly as in the previous case, we denote the reticulated cherry shape as $\{p_xx,p_yp_x,p_yy\}$, and we call $p_x, p_y$ the internal nodes, $x, y$ the terminal nodes, and~$p_yp_x$ the \emph{middle arc}. 
We refer to both cherry and reticulated cherry shape as \emph{shapes}.

A set $P$ of shapes is a \emph{cherry cover} for a binary network $N=(V,A)$, if every arc in $A$, apart from the root arc, is contained in a single shape of $P$.
To adapt the concept of cherry covers to semi-binary networks, we require the following definition. 
The \emph{bulged version} of a semi-binary network $N$, denoted by $B(N)$, is the multigraph obtained from $N$ by replacing the outgoing arc of each reticulation node $u$ with $\indeg u=k$ by $k-1$ parallel arcs. Notice that the bulged version of a binary network is (isomorphic to) the network itself.
Then a \emph{cherry cover} of the bulged version $B(N)$ of a semi-binary network $N$ is a set $P$ of shapes such that each arc of $B(N)$, apart from the root arc, is contained in a single shape of $P$. See an example in Figure \ref{fig.5}.

Cherry covers can be used to characterize the class of tree-based networks.

\begin{thm}[adapted from Theorem 3.3 of \cite{van2021unifying}]\label{thm:TBiffCC}
A semi-binary phylogenetic network $N$ is tree-based if, and only if, $B(N)$ has a cherry cover.
\end{thm}

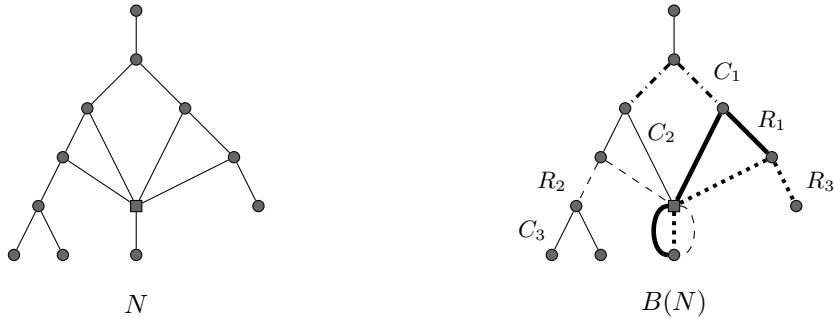
\begin{figure}[h]
\centering

\begin{tikzpicture}[scale=0.65]

\draw (0,1)                 node[mi]    (root)      {};
\draw (0,0)                 node[mi]    (A)      {};
\draw (1,-1)                 node[mi]    (A1)      {};
\draw (-1,-1)                 node[mi]    (A2)      {};
\draw (-1.5,-2)                 node[mi]    (A3)      {};
\draw (-2,-3)                 node[mi]    (A6)      {};
\draw (0,-3)                 node[bx]    (A4)      {};
\draw (2,-2)                 node[mi]    (A5)      {};
\draw (2.5,-3)                 node[mi]    (A7)      {};
\draw (0,-4)                 node[mi]    (A8)      {};
\draw (-2.5,-4)                 node[mi]    (A9)      {};
\draw (-1.5,-4)                 node[mi]    (A10)      {};


\draw[style={-}](root)to(A);
\draw[style={-}](A)to(A1);
\draw[style={-}](A)to(A2);
\draw[style={-}](A3)to(A2);
\draw[style={-}](A4)to(A1);
\draw[style={-}](A5)to(A1);
\draw[style={-}](A2)to(A4);
\draw[style={-}](A3)to(A4);
\draw[style={-}](A3)to(A6);
\draw[style={-}](A5)to(A4);
\draw[style={-}](A5)to(A7);
\draw[style={-}](A4)to(A8);
\draw[style={-}](A6)to(A9);
\draw[style={-}](A6)to(A10);


\draw (7+2+2+0,1)                 node[mi]    (root)      {};
\draw (7+2+2+0,0)                 node[mi]    (A)      {};
\draw (7+2+2+1,-1)                 node[mi]    (A1)      {};
\draw (7+2+2+-1,-1)                 node[mi]    (A2)      {};
\draw (7+2+2+-1.5,-2)                 node[mi]    (A3)      {};
\draw (7+2+2+-2,-3)                 node[mi]    (A6)      {};
\draw (7+2+2+0,-3)                 node[bx]    (A4)      {};
\draw (7+2+2+2,-2)                 node[mi]    (A5)      {};
\draw (7+2+2+2.5,-3)                 node[mi]    (A7)      {};
\draw (7+2+2+0,-4)                 node[mi]    (A8)      {};
\draw (7+2+2+-2.5,-4)                 node[mi]    (A9)      {};
\draw (7+2+2+-1.5,-4)                 node[mi]    (A10)      {};

\node (label) at (0,-5) {$N$};


\draw[style={-}](root)to(A);
\draw[style={-},dashdotted,line width = 1.25](A)to(A1);
\draw[style={-},dashdotted,line width = 1.25](A)to(A2);
\draw[style={-}](A3)to(A2);
\draw[style={-},line width = 1.75](A4)to(A1);
\draw[style={-},line width = 1.75](A5)to(A1);
\draw[style={-}](A2)to(A4);
\draw[style={-},dashed](A3)to(A4);
\draw[style={-},dashed](A3)to(A6);
\draw[style={-},line width = 1.5,dotted](A5)to(A4);
\draw[style={-},line width = 1.5,dotted](A5)to(A7);
\draw[style={-},dashed, in = 0, out = 0](A4)to(A8);
\draw[style={-},line width = 1.75, in = 180, out = 180](A4)to(A8);
\draw[style={-},line width = 1.5,dotted](A4)to(A8);
\draw[style={-}](A6)to(A9);
\draw[style={-}](A6)to(A10);

\node (label) at (11,-5) {$B(N)$};
\node (label) at (11+1.1,-0.25) {\small $C_1$};
\node (label) at (11-2.9,-3.5) {\small $C_3$};
\node (label) at (11-1.75+1.25+0.25,-1.5) {\small $C_2$};
\node (label) at (11+2,-1.25) {\small $R_1$};
\node (label) at (11-2.5,-2.5) {\small $R_2$};
\node (label) at (11+3,-2.5) {\small $R_3$};
\end{tikzpicture}
\caption{\small A semi-binary network $N$ and its bulged version $B(N)$ with a possible cherry cover $\{C_1,C_2,C_3,R_1,R_2,R_3\}$. The different line type indicates the edges of each cherry shape.}
\label{fig.5}
\end{figure}




\subsection{Fence decomposition}\label{subsec:bindecomp}

We include here another static characterization of binary tree-based networks based on an arc partition, called \emph{maximum zig-zag 
trails}~\cite{zhang2016tree, hayamizu2021structure}. 
Let~$N$ be a binary network. A \emph{zig-zag trail} of length~$k$ is a sequence~$(a_1,a_2,\ldots, a_k)$ of arcs where~$k\ge 1$, where either $\head(a_{2i-1}) = \head(a_{2i})$ and $\tail(a_{2j})=\tail(a_{2j+1})$ holds for~$i\in \left[\lfloor \frac{k}{2}\rfloor \right] = \{1,2,\ldots, \lfloor \frac{k}{2}\rfloor \}$ and~$j\in \left[ \lfloor \frac{k-1}{2}\rfloor \right]$, or $\tail(a_{2i-1}) = \tail(a_{2i})$ and $\head(a_{2j}) = \head(a_{2j+1})$ holds for~$i\in \left[ \lfloor \frac{k}{2}\rfloor \right]$ and~$j\in \left[ \lfloor \frac{k-1}{2}\rfloor \right]$. We 
call a zig-zag trail~$Z$ \emph{maximal} if there is no zig-zag trail that 
contains~$Z$ as a subsequence. Depending on the nature of~$\tail(a_1)$ 
and~$\tail(a_k)$, we have four possible maximal zig-zag trails.

\begin{itemize}
	\item \emph{Crowns}: $k\ge 4$ is even and $\tail(a_1) = \tail(a_k)$ or 
	$\head(a_1) = \head(a_k)$.
    By reordering the arcs, assume henceforth that~$\head(a_1) = \head(a_k)$.
	\item \emph{$M$-fences}: $k\ge 2$ is even, it is not a crown, and $\tail(a_i)$ is a tree vertex for every~$i\in[k]$.
	\item \emph{$N$-fences}: $k\ge 1$ is odd and $\tail(a_1)$ or~$\tail(a_k)$, but not both, is 
	a reticulation. By reordering the arcs, assume henceforth that~$\tail(a_1)$ is a reticulation and~$\tail(a_k)$ a tree vertex. 
	\item \emph{$W$-fences}: $k\ge 2$ is even and both~$\tail(a_1)$ 
	and~$\tail(a_k)$ are 
	reticulations.
\end{itemize}
We call a set~$S$ of maximal zig-zag trails a \emph{fence decomposition} 
of~$N$ if the elements of~$S$ partition all arcs, except for the root arc, 
of~$N$. See an example in Figure \ref{fig:bstruct}. One can characterize networks using such a decomposition.

\begin{thm}[adapted from Theorem 4.2 of 
\cite{hayamizu2021structure}]\label{thm:UniqueMaxZigZag}
	Any binary network~$N$ has a unique fence decomposition.
\end{thm}

One can also obtain a forbidden structure characterization of binary tree-based networks using the decomposition.

\begin{lem}[adapted from Corollary 4.6 of \cite{hayamizu2021structure}]\label{lem:TB=NoW}
    Let~$N$ be a binary network. Then~$N$ is tree-based if and only if it has no $W$-fences.
\end{lem}

For non-crown maximal zig-zag trails, we require the notion of \emph{endpoints} to be used in later proofs. 
Let~$A = (a_1,\ldots, a_k)$ be a maximal zig-zag trail.
We shall write~$End(A)$ to denote the set of endpoints of~$A$.
If~$A$ is an $M$-fence, then $End(A) = \{\head(a_1), \head(a_k)\}$.
If~$A$ is an $N$-fence, then $End(A) = \{\tail(a_1), \head(a_k)\}$. 
If~$A$ is a $W$-fence, then $End(A) = \{\tail(a_1), \tail(a_k)\}$.

\begin{figure}[ht]
\centering

\begin{tikzpicture}[scale=1]


\draw (+2-1.5+4-0.5+ 1,  0+0)        node[mini]    (B)          {};
\draw (+2-1.5+4-0.5+ 2,  0+0)        node[mini]    (C)          {};
\draw (+2-1.5+4-0.5+ .5, 0+ -1)        node[minibx]    (A1)          {};
\draw (+2-1.5+4-0.5+ 1.5,0+ -1)        node[minibx]    (B1)          {};

\node (label) at (-0.5+2-2.5+4,-.5) {\small$\bullet$ \textbf{Crown:}};

\draw[style={-}](A1)to(B);
\draw[style={-}](B)to(B1);
\draw[style={-}](B1)to(C);
\draw[style={-}](C)to(A1);

\draw (-0.5+2+4-1+0.25, -.5)        node[label]    (a1)          {\tiny $a_1$};
\draw (-0.5+2+4-1+.75, -.4)        node[label]    (a1)          {\tiny $a_2$};
\draw (-0.5+2+4-1+1.25, -.75)        node[label]    (a1)          {\tiny $a_3$};
\draw (-0.5+2+4-1+0.5, -.85)        node[label]    (a1)          {\tiny $a_4$};



\draw (-0.25+9+1+ 0,  +0)          node[minibx]    (A)          {};
\draw (-0.25+9+1+ 1,  +0)        node[mini]    (B)          {};
\draw (-0.25+9+1+ 2,  +0)        node[minibx]    (C)          {};
\draw (-0.25+9+1+ .5, + -1)        node[minibx]    (A1)          {};
\draw (-0.25+9+1+ 1.5,+ -1)        node[minibx]    (B1)          {};

\node (label) at (2+9+-2.75,-.5) {\small$\bullet$ \textbf{W-fence:}};

\draw[style={-}](A)to(A1);
\draw[style={-}](A1)to(B);
\draw[style={-}](B)to(B1);
\draw[style={-}](B1)to(C);

\draw (1-0.25+9+0.25,    -.5)        node[label]    (a1)          {\tiny $a_1$};
\draw (1-0.25+9+.75,     -.5)        node[label]    (a1)          {\tiny $a_2$};
\draw (1-0.25+9+1.25,    -.5)        node[label]    (a1)          {\tiny $a_3$};
\draw (1-0.25+9+1.75,    -.5)        node[label]    (a1)          {\tiny $a_4$};


\draw (+2-1.5+4-0.5+ 1,  -2+0)        node[mini]    (B)          {};
\draw (+2-1.5+4-0.5+ 2,  -2+0)        node[mini]    (C)          {};
\draw (+2-1.5+4-0.5+ .5, -2+-1)        node[mini]    (A1)          {};
\draw (+2-1.5+4-0.5+ 1.5,-2+ -1)        node[minibx]    (B1)          {};
\draw (+2-1.5+4-0.5+ 2.5,-2+ -1)        node[mini]    (C1)          {};

\node (label) at (1.6+2-1.5+4-0.5+-2.5,-2-.5) {\small$\bullet$ \textbf{M-fence:}};

\draw[style={-}](A1)to(B);
\draw[style={-}](B)to(B1);
\draw[style={-}](B1)to(C);
\draw[style={-}](C)to(C1);

\draw (+2-1.5+4-0.5++.75, -2-.5)        node[label]    (a1)          {\tiny $a_1$};
\draw (+2-1.5+4-0.5++1.25,-2 -.5)        node[label]    (a1)          {\tiny $a_2$};
\draw (+2-1.5+4-0.5++1.75,-2 -.5)        node[label]    (a1)          {\tiny $a_3$};
\draw (+2-1.5+4-0.5++2.25,-2 -.5)        node[label]    (a1)          {\tiny $a_4$};



\draw (-0.25+9+1+ 0,  -2+0)          node[minibx]    (A)          {};
\draw (-0.25+9+1+ 1,  -2+0)        node[mini]    (B)          {};
\draw (-0.25+9+1+ .5, -2+-1)        node[minibx]    (A1)          {};
\draw (-0.25+9+1+ 1.5,-2+ -1)        node[mini]    (B1)          {};

\node (label) at (-5+0.25-0.25+9+1-0.25+1+5+-2.5,-2-.5) {\small$\bullet$ \textbf{N-fence:}};

\draw[style={-}](A)to(A1);
\draw[style={-}](A1)to(B);
\draw[style={-}](B)to(B1);

\draw (-0.25+9+1+0.25,-2 -.5)        node[label]    (a1)          {\tiny $a_1$};
\draw (-0.25+9+1+.75,-2-.5)        node[label]    (a1)          {\tiny $a_2$};
\draw (-0.25+9+1+1.25, -2-.5)        node[label]    (a1)          {\tiny $a_3$};



\draw (7,-2-3-0.25)          node[mi]    (A)          {};
\draw (7,-2-3+0.5)          node[mi]    (root)          {};

\draw (5    ,-2+-4)            node[mi]    (A1)          {};
\draw (3.5  ,-2+-5.5)            node[mi]    (A2)          {};
\draw (6    ,-2+-4.5)            node[mi]    (A3)          {};
\draw (5    ,-2+-5.5)            node[mi]    (A4)          {};
\draw (7    ,-2+-5.5)            node[mi]    (A5)          {};
\draw (5-0.5,-2+-6.5)        node[bx]    (A6)          {};
\draw (7-0.5,-2+-6.5)        node[bx]    (A7)          {};
\draw (2.5  ,-2+-8)            node[mi]    (A8)          {};
\draw (4    ,-2+-7.25)            node[bx]    (A9)          {};
\draw (4    ,-2+-8)            node[mi]    (A12)          {};
\draw (9    ,-2+-4)            node[mi]    (B1)          {};
\draw (8    ,-2+-5.5)            node[mi]    (B2)          {};
\draw (10   ,-2+-5.5)           node[mi]    (B3)          {};
\draw (10 -0.25+0.1 -2-0.25,-2+-5.5-0.5)           node[mi]    (B11)          {};
\draw (10 -0.25 -2-0.25+2,-2+-8)           node[mi]    (B12)          {};
\draw (7.5  ,-2+-6.5)            node[mi]    (B4)          {};
\draw (8.5  ,-2+-6.5)            node[bx]    (B5)          {};
\draw (10.5 ,-2+-8)           node[mi]    (B6)          {};
\draw (7    ,-2+-7.25)            node[bx]    (B7)          {};
\draw (8    ,-2+-7.25)            node[bx]    (B8)          {};
\draw (7    ,-2+-8)            node[mi]    (B9)          {};
\draw (8    ,-2+-8)            node[mi]    (B10)          {};

\draw[style={-}](A)to(root);
\draw[style={-},dotted,line width = 1.75](A)to(A1);
\draw[style={-},line width = 1.75](A1)to(A2);
\draw[style={-},line width = 1.75](A3)to(A1);
\draw[style={-}](A3)to(A4);
\draw[style={-}](A3)to(A5);
\draw[style={-},line width = 1.75](A6)to(A4);
\draw[style={-},line width = 1.75](A7)to(A5);
\draw[style={-},line width = 1.75](A7)to(A4);
\draw[style={-},line width = 1.75](A6)to(A5);
\draw[style={-},dashdotted,line width = 1.5](A2)to(A8);
\draw[style={-},dashdotted,line width = 1.5](A2)to(A9);
\draw[style={-},dashdotted,line width = 1.5](A6)to(A9);
\draw[style={-}](A9)to(A12);
\draw[style={-},dotted,line width = 1.75](A)to(B1);
\draw[style={-}](B2)to(B1);
\draw[style={-},dotted,line width = 1.25](B2)to(B5);
\draw[style={-},dotted,line width = 1.25](B2)to(B11);
\draw[style={-}](B3)to(B1);
\draw[style={-}](B11)to(B12);
\draw[style={-}](B11)to(B4);
\draw[style={-},dotted,line width = 1.25](B3)to(B5);
\draw[style={-},dotted,line width = 1.25](B3)to(B6);
\draw[style={-},dashed,line width = 1.25](B4)to(B7);
\draw[style={-},dashed,line width = 1.25](B4)to(B8);
\draw[style={-},dashed,line width = 1.25](A7)to(B7);
\draw[style={-},dashed,line width = 1.25](B5)to(B8);
\draw[style={-}](B10)to(B8);
\draw[style={-}](B9)to(B7);

\node (label) at (7,-11) {\LARGE $N$};
\node (label) at (6.5,1.5) {\textbf{Binary structures}};

\node (label) at (5.75,-5.25) { $\hat{M_1}$};
\node (label) at (5.75-1.9,-6.55) { $\hat{M_2}$};
\node (label) at (5.75+3+1.25,-6.75)                        { $\hat{M}_3$};
\node (label) at (5.75-1.9+2+1.1,-7)                        { $\hat{M}_4$};
\node (label) at (5.75+3+1.25+0.75-1-1,-6.75-1.1)           { $\hat{M}_5$};
\node (label) at (5.75+3+1.25+0.75-1-1+.5,-6.75-1.1-1.5)    { $\hat{M}_6$};
\node (label) at (5.75-1.9+2+1.1-1.25+0.25,-7.5)            { $\hat{C}_1$};
\node (label) at (5.75-1.9+2+1.1-1.25+0.75+0.25-0.1,-9.6)   { $\hat{N}_3$};
\node (label) at (5.75-1.9+2+1.1-1.25+0.75+1+0.9,-9.6)      { $\hat{N}_4$};
\node (label) at (5.75-1.9+2+1.1-1.25+0.75,-9)              { $\hat{W}_1$};
\node (label) at (5.75-1.9+2+1.1-1.25-3,-8.5)               { $\hat{N}_1$};
\node (label) at (5.75-1.9+2+1.1-1.25-3+1-0.1,-9.6)         { $\hat{N}_2$};
\end{tikzpicture}
\caption{\small Representation of the different zig-zag trails present in the binary setting and a network $N$ with fence decomposition $S = \cup_{i=1}^6\hat{M}_i\cup_{i=1}^4\hat{N}_i\cup \hat{W}_1\cup \hat{C}_1$. The different line types indicates the arcs of each maximal zig-zag trail.}
\label{fig:bstruct}
\end{figure}

\section{Cherry covers and support trees in semi-binary networks}
\label{sec:cv_suppT}

Recall that cherry covers in the semi-binary setting were defined over the bulged version of the network, obtained by adding, between each reticulation and its child, as many parallel arcs as the indegree of the reticulation considered minus 1. Over this bulged version, the definition of a cherry cover is given as in the binary case. 
These additional parallel arcs may present duplicates when counting the number of cherry covers for a semi-binary network. 
For example, consider two reticulated cherry shapes~$\{a_1,b,c\}$ and~$\{a_2,b,c\}$, where~$a_1$ and~$a_2$ are parallel arcs. 
Clearly, one should consider the shapes to be equivalent; in this vein, we say that such two reticulated cherry shapes are \emph{isomorphic}.
Additionally, we say that two cherry covers $P_1$ and $P_2$ are \emph{isomorphic} if there exists a one-to-one correspondence between their set of reticulated cherry shapes into isomorphic reticulated cherry shapes. Informally, the isomorphism between two cherry covers can be understood as an isomorphism between the parallel arcs added in the bulged version (in a sense that we do not distinguish them). 

We denote by $\mathcal{P}_N$ and $\mathcal{S}_N$ the set of all (non-isomorphic) cherry covers and the set of support trees of a semi-binary network~$N$, respectively. Notice that from Theorem~\ref{thm:TBiffCC}, $\mathcal{P}_{N} \neq \emptyset$ if and only if $N$ is tree-based.  Moreover, by definition, $N$ is tree-based if and only if $\mathcal{S}_{N} \neq \emptyset$. 
In the next proposition we prove that the number of cherry covers coincides with the number of support trees in semi-binary networks.

\begin{prop}\label{prop:CC=ST-SB}
Let~$N$ be a semi-binary network. Then~$|\mathcal{P}_N| = |\mathcal{S}_N|$. 
\end{prop}
\begin{proof}
    If~$N$ is not tree-based, then~$|\mathcal{P}_N| = |\mathcal{S}_N| = 0$. So suppose that~$N$ is tree-based.
    
    Let~$P$ be a cherry cover of~$N$. 
    Let $\{p_xx,p_yp_x,p_yy\}\in P$ be a reticulated cherry shape. 
    Observe that deleting the middle arc $p_yp_x$ reduces the indegree of~$p_x$ by 1 and the outdegree of~$p_y$ by 1. The degrees of all other vertices remain unchanged. Therefore, the resulting subgraph has the same leafset as that of~$N$.
    Consider $T$, a tree obtained by removing all reticulation arcs of~$N$ that are covered as middle arcs of reticulated cherry shapes in~$P$.
    By the same argument as above,~$T$ is a subgraph of~$N$ with the same leafset as that of~$N$.  
    Thus,~$T$ must be a support tree. So for every cherry cover, we can construct a support tree. 
    This shows that~$|\mathcal{P}_N|\le |\mathcal{S}_N|$.

    Let~$S$ be a support tree of~$N$, obtained by removing all but one incoming reticulation arcs from each reticulation.
    We cover each such deleted arcs~$\{a_1,\ldots, a_k\}$ with the middle arc of a reticulated cherry shape.
    We claim that we can extend the covering to a cherry cover of~$N$.
    Since~$N$ is semi-binary,~$\tail(a_i)$ is a tree vertex of outdegree-2 for each~$i$.
    So covering~$a_i$ as a middle arc also covers the other outgoing arc of~$\tail(a_i)$ and also the outgoing arc of~$\head(a_i)$.
    Covering in this way does not cause any overlaps between the reticulated cherry shapes, except possibly at the parallel outgoing arcs of reticulations (in the bulged version of~$N$).
    Indeed, it is not possible for both outgoing arcs of a tree vertex to be middle arcs, since~$S$ is a support tree of~$N$.
    At this moment, every outgoing parallel arcs of reticulations are covered (in the bulged version of~$N$).
    Let~$u$ be a tree vertex.
    Either both outgoing arcs of~$u$ are covered by a single reticulated cherry shape, or none of its outgoing arcs are covered.
    In the latter case, we cover the two outgoing arcs using a cherry shape.
    Repeating this ensures that all outgoing arcs of all vertices are covered by cherry shapes and reticulated cherry shapes.
    So for every support tree, we can construct a cherry cover.
    It follows that~$|\mathcal{S}_N|\le |\mathcal{P}_N|$.
\end{proof}

\section{Covering fences using cherry covers}
\label{sec:fences_vs_cv}

In this section, we show how one can obtain a cherry cover from a fence decomposition in binary tree-based networks. 
Recall that a network is tree-based if and only if it has a cherry cover (Theorem~\ref{thm:TBiffCC}). Additionally, a network is tree-based if and only if has no $W$-fences in its unique fence decomposition (Lemma~\ref{lem:TB=NoW}).
Note that, in a binary network, exactly one reticulation arc for every reticulation is contained in a reticulated cherry shape as a middle arc.

Let~$N$ be a network and let~$\hat{N}$ be an $N$-fence of~$N$. In what follows, we shall write~$\hat{N} := (a_1,a_2,\ldots, a_{k})$, and we will let~$c_{2j-1}$ denote the arc connecting~$\head(a_{2j-1})$ and its only child for~$j\in\left[\frac{k-1}{2}\right]$. 
We refer to~$a_1$ as the \emph{first arc} of the $N$-fence~$\hat{N}$.
The first lemma states that although a tree-based network may have non-unique cherry covers, the reticulated cherry shapes that cover arcs of $N$-fences are fixed.

\begin{lem} [adapted from Lemma 13 of \cite{van2023making}]\label{lem:n_fence_unique}
    Let~$N$ be a binary tree-based network that has an $N$-fence~$\hat{N}$ of length at least 3. 
    Then every cherry cover of~$N$ contains the reticulated cherry shapes $\{c_{2j-1}, a_{2j},a_{2j+1}\}$ for~$j\in\left[\frac{k-1}{2}\right]$.
\end{lem}

Note that the arcs of the reticulated cherry shapes do not exactly coincide with the arcs of~$\hat{N}$ in Lemma~\ref{lem:n_fence_unique}.
In particular, the arc~$a_1$ of~$\hat{N}$ is not covered; we shall see later that~$a_1$ will be covered as an outgoing arc of a reticulation in another fence.
Also, observe that the reticulated cherry shapes cover the arcs~$c_{2j-1}$, which are not in~$\hat{N}$. 
Such arcs are contained as first arcs in other $N$-fences. 
They will not be covered in other reticulated cherry shapes, similarly to how~$a_1$ was not covered by the reticulated cherry shapes corresponding to~$\hat{N}$.

We shall show similar results for $M$-fences and crowns. Let $\hat{M}$ be an $M$-fence of~$N$. 
In what follows, we shall write~$\hat{M} := (a_1,a_2,\ldots, a_{k})$, and we will let~$c_{2j}$ denote the arc connecting~$\head(a_{2j})$ and its child for~$j\in\left[\frac{k-2}{2}\right]$.

\begin{lem}\label{lem:M-fenceSizeCC}
    Let~$N$ be a binary tree-based network, and let~$\hat{M}$ be an $M$-fence of length~$k\ge 4$. 
    There exists $k/2$ distinct ways of covering the arcs of~$\hat{M}$ and the arcs~$c_{2j}$ for~$j\in\left[\frac{k-2}{2}\right]$ using cherry and reticulated cherry shapes.
\end{lem}
\begin{pf}
    In any cherry cover of~$N$, exactly one incoming arc of every reticulation is covered as a middle arc of a reticulated cherry shape. 
    Since the $M$-fence contains $k/2-1$ reticulations, we must have exactly $k/2-1$ reticulated cherry shapes that cover the arcs of~$\hat{M}$ and the arcs $c_{2i}$.
    By construction, the remaining two arcs of~$\hat{M}$ must be arcs~$a_m$ and~$a_n$ for some $m,n\in[k]$.
    Suppose without loss of generality that~$n>m$. 
    We show that~$n = m+1$.
    If not, then consider the arc~$a$ that is different from $a_m$ in~$\hat{M}$ where~$\tail(a) = \tail(a_m)$.
    Note that~$a=a_{m-1}$ or~$a=a_{m+1}.$ By assumption,~$a$ is covered by a reticulated cherry shape. 
    But since~$\tail(a)$ is a tree vertex, and since~$a$ is covered as the middle arc of a reticulated cherry shape,~$a$ can only be covered in a shape together with~$a_m$.
    This gives the required contradiction, since~$a_m$ is yet to be covered.
    Thus~$n=m+1$, and we may cover these remaining arcs as a cherry shape, $\{a_m,a_{m+1}\}$.

    Therefore every $M$-fence of length~$k\ge4$ must be covered by a single cherry shape and~$k/2-1$ reticulated cherry shapes in every cherry cover.
    Note that the choice of arcs in the cherry shape determines the cherry cover for the $M$-fence subgraph.
    We conclude that there are $k/2$ distinct ways of covering the arcs of~$\hat{M}$ and the arcs~$c_{2j}$ for~$j\in\left[\frac{k-2}{2}\right]$ using cherry and reticulated cherry shapes.
\end{pf}

Finally let $\hat{C}$ be a crown of~$N$. 
In what follows, we shall write~$\hat{C} := (a_1,a_2,\ldots, a_{k})$, and we will let~$c_{2j}$ denote the arc connecting~$\head(a_{2j})$ and its child for~$j\in\left[\frac{k}{2}\right]$. 

\begin{lem}\label{lem:CrownSizeCC}
    Let~$N$ be a binary tree-based network, and let~$\hat{C}$ be a crown of length~$k\ge 4$. There exists exactly 2 distinct ways of covering the arcs of~$\hat{C}$ and the arcs~$c_{2j}$ for~$j\in\left[\frac{k}{2}\right]$  using reticulated cherry shapes.
\end{lem}

\begin{pf}
    As in the proofs of the previous two lemmas, we use the fact that exactly one incoming arc of every reticulation is covered as a middle arc of a reticulated cherry shape.
    Each reticulated cherry shape consists of 3 arcs. 
    The crown, together with the arcs~$c_{2j}$, has exactly $3k/2$ arcs in total; the crown contains exactly $k/2$ reticulations.
    It follows that every arc of~$\hat{C}$ and the arcs $c_{2j}$ for each $j$ must be covered by a reticulated cherry shape. 
    There are two ways to do this, by considering either of the covers
 \[\{c_{2j},a_{2j+1}, a_{2j+2}: j\in \left[\frac{k-2}{2} \right]\} \cup \{c_k, a_1, a_2\}\]

    or
    \[\{\{c_{2j},a_{2j}, a_{2j-1}\}: j\in[k/2]\}.\]

    No other ways of using reticulated cherry shapes to cover the $3k/2$ arcs exist.
    Therefore, there are exactly 2 distinct ways of covering the arcs of~$\hat{C}$ and the arcs~$c_{2j}$ for~$j\in\left[\frac{k}{2}\right]$.
\end{pf}

\begin{thm}\label{thm:NumberOfCherryCovers}
    Let~$N$ be a binary tree-based network. 
    Let~$\hat{M_1},\hat{M_2},\ldots, \hat{M_m}$ denote the $M$-fences of~$N$ of length $k_i\ge 4$ where~$i\in[m]$ and let~$\hat{C_1},\hat{C_2},\ldots,\hat{C_c}$ denote the crowns of~$N$.
    Then~$N$ has $2^{c-m} k_1 k_2 \cdots k_m$ distinct cherry covers and support trees.
\end{thm}
\begin{pf}
    We can partition every cherry cover so that each block corresponds to a maximum zig-zag trail, together with outgoing arcs of the reticulations of the trail. 
    This follows from the fact that every arc with a reticulation tail must be covered by a reticulated cherry shape.
    By Lemmas~\ref{lem:n_fence_unique}, \ref{lem:M-fenceSizeCC}, and~\ref{lem:CrownSizeCC}, we know that each $N$-fence of length at least 3, $M$-fence of length~$k$ at least 4, and crown (together with outgoing arcs of reticulations in each zig-zag trail) can be covered in 1, $k/2$, and~$2$ ways. 
    The permutations are counted locally within each zig-zag trail, and they have no effect on one another.
    It follows that there are 
    \[2^c \times {\displaystyle \prod_{i=1}^m \frac{k_i}{2}} = 2^{c-m}k_1k_2\cdots k_m\]
    many possible cherry covers.
    By \Cref{prop:CC=ST-SB}, this is also the number of support trees for~$N$.
\end{pf}

\section{Extending fences to non-binary networks}
\label{sec:extension}


The goal of this section is to generalize the existence of a unique fence decomposition from binary networks (Theorem \ref{thm:UniqueMaxZigZag}) to non-binary networks.  

We start by extending the definitions of fences and crowns, introduced in Subsection \ref{subsec:bindecomp}, to non-binary networks. 
In extending crowns, we shall glue together a crown to a $W$-fence, an $N$-fence, or to another crown.
Roughly speaking, one can view non-binary crowns as structures that contain a binary crown.
In particular, a binary crown will also be considered to be a non-binary crown.
For the $W,N,M$-fences, we impose that they do not contain any binary crowns. 
By the properties of binary networks, it is clear that we cannot have a zig-zag trail containing a crown unless the whole trail forms a crown. Thus, the extra condition of acyclic paths over the $W,N,M$-fences for the non-binary case is not a more restrictive condition on the binary case. 

Let $N$ be non-binary; we introduce five more zig-zag trails. Let~$ (a_1,a_2,\ldots,a_k,c_1,c_2,\ldots,c_{2k'})$ be a zig-zag trail, of different arcs, where $(c_1,\ldots,c_{2k'})$ is a crown. We say that such a path is a

\begin{itemize}
    \item[$\bullet$] \emph{Lower (Upper) $W$-crown}: $k\geq 1$ is odd ($k\geq 2$ is even), $(a_1,\ldots,a_k)$ does not contain a crown, $head(a_k) = head(c_1)$ ($tail(a_k) = tail(c_1)$) and $tail(a_1)$ is a reticulation.
    \item[$\bullet$] \emph{Lower (Upper) $N$-crown}: $k\geq 2$ is even ($k\geq 1$ is odd), $(a_1,\ldots,a_k)$ does not contain a crown, $head(a_k) = head(c_1)$ ($tail(a_k) = tail(c_1)$) and $head(a_1)$ is a tree node.
    \item[$\bullet$] \emph{Double-crown:} There exists a $j<k$ such that $(a_1,a_2,\ldots,a_j)$ forms a crown and $(a_{j+1},\ldots,a_k)$ does not contain a crown.
\end{itemize}

\input{NBStructures}

As in the binary case, we define endpoints of the new structures. 
Formally, if $A$ is an upper/lower $W$-crown, then $End(A) = \{\tail(a_1)\}$, and if $A$ is an upper/lower $N$-crown, then $End(A) = \{\head(a_1)\}$. 
Figure \ref{fig:nbstruct} shows some examples of all the different types of zig-zag trails that we can find in non-binary networks.

We now extend the notion of zig-zag trails to non-binary networks. We use maximum zig-zag trails as defined in Section \ref{subsec:bindecomp}; the only difference between maximum zig-zag trails in binary and those in non-binary networks
is that two maximum zig-zag trails can overlap at arcs in the latter.
Let~$N$ be a non-binary network, and let~$A,B$ be two maximum zig-zag trails of~$N$.
We say that~$A$ and~$B$ \emph{intersect} if they share at least one common arc. Let~$\mathcal{F}(N)$ denote the set of all maximum zig-zag trails of~$N$. 
We call a subset~$F$ of~$\mathcal{F}(N)$ an \emph{extended zig-zag trail} if for any element~$A\in F$, either~$F = \{A\}$, or for any other element~$B\in F$, there exists a sequence~$A = A_1,\ldots, A_k = B$ of zig-zag trails in $F$, such that~$A_i$ intersects with~$A_{i+1}$ for~$i\in[k-1]$.
Such a set~$F$ is called \emph{maximal} if there is no subset~$G$ of~$\mathcal{F}(N)\setminus F$ where~$F\subseteq G$.
We call the set of maximal extended zig-zag trails an \emph{extended fence decomposition}.

\input{Extfence}

Let $F$ be an extended zig-zag trail. We denote by $V(F)$ and $A(F)$ the set of vertices and arcs of $F$, respectively. By considering an extended fence decomposition $D$ of the network, it is clear that each arc must be covered by exactly one maximal extended zig-zag trail $F\in D$. Let $G\subseteq F$. 
We say that $G$ \emph{covers} $F$, or that $G$ \emph{is a covering} of $F$, if $A(G) = A(F)$. Figure \ref{fig.6} shows an example of an extended zig-zag trail and one possible covering.

Additionally, let $F= \{A_1,\ldots,A_k,C_1,\ldots,C_\ell\}$, where $C_i$ represents all the crowns and double-crowns in $F$, and $A_i$ represent the other trails. We denote by $End(F)$ the set of endpoints of each $A_i$. We call this set the set of  \emph{endpoints} of $F$. 

\begin{prop}
    \label{prop:extended_decomp}
    Any non-binary network~$N$ has a unique extended fence decomposition.
\end{prop}
\begin{proof}
    Every arc of~$N$ is covered by a trail in exactly one maximal extended zig-zag trail. Indeed, every maximal zig-zag trail that contains a certain arc is placed in the same  maximal extended zig-zag trail, by definition.
\end{proof}

By construction, we have that extended zig-zag trails are connected, so we have zig-zag trails connecting any two nodes of the trail. The following result states that we do not only have paths between any two nodes of the fence, but we have zig-zag trails, which can be extended to maximal zig-zag trails.

\begin{lem}\label{lem:endp}  
    Let $N$ be a non-binary network and let $F$ be a maximal extended zig-zag trail in $N$ and $u,v\in V(F)$. Then there exists a zig-zag trail $C$ (not necessarily maximal) with all its arcs in $A(F)$, such that $End(C) = \{u,v\}$.
\end{lem}

\begin{pf}
Let $A$ and $B$ be two maximal zig-zag trails in $F$ such that $u$ is a node of $A$ and $v$ is a node of $B$. By construction of $F$, there exists a sequence $A = A_1,A_2,\ldots,A_k = B$ of maximal zig-zag trails of $F$, with $k\geq 2$, such that $A_i$ and $A_{i+1}$ are intersecting for every $i$. 
We are going to construct a maximal zig-zag trail with arcs in $A(F)$, including $u$ and $v$.
Let $d_i$ denote one of the arcs in the intersection of $A_i$ and $A_{i+1}$. Now, we define maximal zig-zag trails $C_i$ as follows:

\begin{itemize}
        \item[(1)] $C_1 = A_1$.
        \item[(2)] For $1<i <k$:
        \begin{itemize}
            \item[-] Denote $C_{i-1} = (c^{i-1}_1,\ldots,c^{i-1}_{m_{i-1}})$ and $A_{i} = (a^i_1,\ldots,a^i_{n_i})$ and let $j\in [m_{i-1}]$ such that $c_j^{i-1} = a_l^i$ for some $l\in [n_i]$. 
            Suppose that $u\in \{\tail(c_h^{i-1}),\head(c_h^{i-1})\}_{h=1}^j$ and $d_i\in\{a_{h}^i\}_{h=l}^{n_i}$, otherwise, rearrange the trails to be in this situation. Then, we define            $$C_i = 
            (c_1^{i-1},\ldots,c_j^{i-1},a_{l+1}^i,a_{l+2}^i,\ldots, a_{n_i}^i).$$
        \end{itemize}
                \item[(3)] Denote $C_{k-1} = (c^{k-1}_1,\ldots,c^{k-1}_{m_{k-1}})$. We have
                $d_{k-1}\in A(C_{k-1})\cap A(A_k)$,
                so let us denote $d_k = c_{p}^{k-1} = a_{q}^k$. Assume that $u\in \{\tail(c_h^{k-1}),\head(c_h^{k-1})\}_{h=1}^p$ and $v\in \{\tail(a_h^{k}),\head(a_h^{k})\}_{h=q}^{n_k}$, otherwise, rearrange the trails. Thus, define 

        $$C_k = \{c_1^{k-1},\ldots,c_p^{k-1},a_{q+1}^k,\ldots,a_{n_k}^k)$$
\end{itemize}

Clearly, each $C_i$ is a maximal zig-zag trail containing $u$ that intersects $A_{i+1}$. Additionally, we can easily see the maximality of $C_i$ from the maximality of $C_{i-1}$ and $A_{i}$. Thus, $C_k$ is a maximal zig-zag trail containing $u$ and $v$ with all its arcs in $F$. From $C_k$, we can extract a zig-zag trail from $u$ to $v$ with all its arcs in $F$.
\end{pf}

It is worth noting that, from the construction given in the proof of Lemma \ref{lem:endp}, we can deduce that, if two disjoint fences $A$ and $B$ are present within the same extended zig-zag trail $F$, we can identify a new fence $C$ intersecting $A$ and $B$, with any two selected vertices from $V(A)$ and $V(B)$. Furthermore, if we select $u$ and $v$ to be endpoints of $F$, Lemma \ref{lem:endp} yields the following result. 

\begin{cor}\label{cor:endp}
    Let $N$ be a non-binary network and let $F$ be a maximal extended zig-zag trail in $N$ and $u,v\in End(F)$. Then, there exists $C\in F$ with $End(C) = \{u,v\}$.
\end{cor}

\section{Characterization of semi-binary tree-based networks}
\label{subsec:CharacterizationSBTBN}

In this section, we extend the results of counting the number of cherry covers (and then of support trees), Theorem \ref{thm:NumberOfCherryCovers}, and the characterization of tree-based networks in terms of forbidden structures, Lemma \ref{lem:TB=NoW}, from binary to semi-binary tree-based networks. 


Proposition \ref{prop:extended_decomp} asserts the existence of a unique extended fence decomposition for any non-binary network. Consequently, this decomposition also applies to semi-binary networks.

Now, our focus is to study the forbidden structures in semi-binary tree-based networks with respect to their extended fence decomposition. Due to the semi-binary property, in particular since the outdegree of tree nodes is 2, it is clear that the network cannot exhibit upper $N$-crowns or upper $W$-crowns. Additionally, we prove the absence of $W$-fences, lower $W$-crowns and double-crowns.

\begin{lem}\label{lem:SBTB-forbidden}
        Let $N$ be a semi-binary tree-based network, and let $D$ be its extended fence decomposition. For all $F\in D$, $F$ satisfies the following conditions:
        
        \begin{itemize}
            \item[\textbf{a)}] $F$ does not contain any $W$-fences.
            \item[\textbf{b)}] $F$ does not contain any lower $W$-crowns.
            \item[\textbf{c)}] $F$ has at most $1$ crown. In particular, $F$ does not contain any double-crowns.
        \end{itemize}
        
\end{lem}

\begin{proof}
Let us start by proving \textbf{a)}. Let $N$ be a semi-binary tree-based network, and suppose there exists a $W$-fence $(a_1,\ldots, a_{2l})$ in a maximal zig-zag trail $F\in D$. Since $N$ is tree-based, it has some underlying support tree $T$ as the result of iteratively deleting all but one incoming arc of each reticulation, ensuring that no new leaves are added at any step.
    
    Observe that $\tail(a_1)$ and $\head(a_1)$ are two connected reticulations, then  $a_1\in E(T)$. Consequently, $a_2\notin E(T)$. As $\tail(a_2)$ is a tree node with outdegree $2$, we have $a_3\in E(T)$; otherwise $\tail(a_2)$ corresponds to a (different) leaf on $T$. Repeating this reasoning, we conclude that $a_{2l-1}\in E(T)$ and $a_{2l}\notin E(T)$. This selection of arcs implies that the reticulation $\tail(a_{2l})$ corresponds to a (different) leaf in the support tree which is a contradiction.

An analogous constructive argument can be used to prove \textbf{b)}.

    Finally, let us show condition \textbf{c)}. Let $N$ be a semi-binary tree-based network, and suppose there exists an extended zig-zag trail $F\in D$ with more than $1$ crown. Let $(a_1,\ldots,a_n)$ and $(b_1,\ldots,b_{m})$ be two different crowns. Notice that $n$ and $m$ are even integers. Let us assume that $\tail(a_1)=\tail(a_n)$ and $\tail(b_1)=\tail(b_m)$, if not, we can rename the arcs to make sure that this is the case. We distinguish the cases where both paths either intersect in a reticulation or they do not. Let $T$ be a support tree of $N$.
    
    Firstly, assume they intersect in a reticulation. Then there exists  $i\in [n]$ and $j\in [m]$ such that $\head(a_i)=\head(a_{i+1}) = \head(b_{j}) = \head(b_{j+1})$. Observe that either $a_i\notin E(T)$ and $a_{i+1}\notin E(T)$, or $b_j\notin E(T)$ and $b_{j+1}\notin E(T)$. Without loss of generality, assume that $a_i\notin E(T)$ and $a_{i+1}\notin E(T)$. Since $a_{i}\notin E(T)$, by the semi-binary property, we have $a_{i-1}\in E(T)$, implying $a_{i-2}\notin E(T)$, and this process can be iterated.
    Repeating the same argument for $a_{i+1}\notin E(T)$, we get $a_{i+2}\in E(T)$, $a_{i+3} \notin E(T)$, etc.    
    Since $n$ is even, we end up having $a_1\notin E(T)$ and $a_n\notin E(T)$. However, due to $\tail(a_1)=\tail(a_n)$, this makes $\tail(a_1)$ a leaf of $T$, contradicting the fact that it is a support tree of $N$.
    
    Now, assume that we have two 
    binary crowns
    $(a_1,\ldots,a_n)$ and $(b_1,\ldots,b_m)$ that are disjoint in their reticulations. By Lemma \ref{lem:endp}, there exists a zig-zag trail $(c_1,\ldots,c_l)$ that joins both structures, such that $c_i \ne a_j$ and $c_i \ne b_k$ for any possible index combinations~$i,j,k$. 
    Due to the semi-binary condition, $\head(c_1)$ and $\head(c_l)$ are reticulations, then $l$ is an even number. By relabelling arcs, we can consider the zig-zag trail $(a_2,a_3,\ldots,a_n,a_1,c_1,\ldots,c_l,b_2,b_3,\ldots,b_m,b_1)$ where $\head(a_1) = \head(a_2) = \head(c_1)$, $\head(b_1) = \head(b_{2}) = \head(c_l)$, $\tail(a_1)=\tail(a_n)$ and $\tail(b_1)=\tail(b_m)$. 
    
    If $a_1\notin E(T)$ and $a_{2}\notin E(T)$, then applying the same argument as in the previous case we reach a contradiction. Therefore, we either have $a_1\in E(T)$ or $a_{2}\in E(T)$. In that case, $c_1\notin E(T)$, and since $l$ is even, $c_{l}\in E(T)$, implying $b_1\notin E(T)$ and $b_2\notin E(T)$. Again, as in the previous case this leads to a contradiction.

Finally, since (by definition) double-crowns contain two binary crowns, then $F$ does not contain those structures. 
\end{proof}

Lemma \ref{lem:SBTB-forbidden} restrict the possible structures that we can find in the fence decomposition of a tree-based network, so we can classify the maximal extended zig-zag trails that we can find in semi-binary tree-based networks.

\begin{defn}
    Let $N$ be a semi-binary tree-based network, $D$ its extended fence decomposition and let $F\in D$. We say that:

    \begin{itemize}
        \item[ - ] $F$ is a \emph{generalized $N$-fence} if $F$ does not have crowns and $End(F)$ contains exactly one reticulation.
        \item[ - ] $F$ is a \emph{generalized $M$-fence} if $F$ does not have crowns and $End(F)$ contains only tree vertices.
        \item[ - ] $F$ is a \emph{generalized crown} if $F$ has one crown and $End(F)$ is either empty or it only contains tree vertices.
    \end{itemize}
\end{defn}

A reinterpretation of Lemma \ref{lem:SBTB-forbidden}
shows that the previous definitions classify the maximal extended zig-zag trails that we can find in semi-binary tree-based networks.

\begin{cor}\label{cor:struct}
    Let $N$ be a semi-binary tree-based network, $D$ its extended fence decomposition and let $F\in D$. Then, one of the following is satisfied:

    \begin{itemize}
        \item[\textbf{(1)}] $F$ is a generalized $N$-fence.
        \item[\textbf{(2)}] $F$ is a generalized $M$-fence.
        \item[\textbf{(3)}] $F$ is a generalized crown.
    \end{itemize}
\end{cor}

\begin{pf}
     By Lemma \ref{lem:SBTB-forbidden} \textbf{c)}, we have that $F$ either has one crown or no crowns. Now, suppose that $End(F)$ has at least $2$ reticulations $u$ and $v$. By Corollary \ref{cor:endp}, we have a maximal zig-zag trail $(a_1,\ldots,a_k)$ that goes from $u$ to $v$. If this trail does not contain crowns, then it is a maximal $W$-fence, which, by Lemma \ref{lem:SBTB-forbidden} \textbf{a)}, we know that are prohibited in tree-based networks. If the trail has a crown, we can consider $k'\in [k]$ to be the smallest integer such that $(a_i,a_{i+1},\ldots,a_{2k'})$ is a crown for some $i$. Thus, the zig-zag trail $(a_1,\ldots, a_{2k'})$ is a maximal $W$-crown, which, by Lemma \ref{lem:SBTB-forbidden} \textbf{b)} are also prohibited. If $F$ has a crown and $End(F)$ has one reticulation, a similar argument leads us to $F$ having a maximal $W$-crown. In other case, $F$ satisfy \textbf{(1)}, \textbf{(2)} or \textbf{(3)}.
\end{pf}

Now that we have described the structures that appear in semi-binary tree-based networks, our objective is to characterize tree-based networks and to count support trees. Recall that, to construct a support tree $T$, we have to select one reticulation arc from each reticulated cherry shape of the network, in a way that we do not create new leaves. Since we have to make a selection of arcs, we would not need the whole information of each extended zig-zag trail. Instead, we will work with covers. The following results will give us minimal coverings (w.r.t. the inclusion), for each extended zig-zag trail, that will be useful to count support trees in the latter.

\begin{cor}\label{cor:coverN}
    Let $F$ be a generalized $N$-fence with $End(F) = \{u,v_1,\ldots,v_k\}$, where $u$ is the only reticulation endpoint and $v_i$ are tree-vertices. Let $G = \bigcup_{i=1}^k \hat{N}_i$, where $\hat{N}_i$ are maximal $N$-fences in $F$ going from $u$ to $v_i$. Then, $G$ is a covering of $F$. 
\end{cor}

\begin{pf}
 By Corollary \ref{cor:endp}, there exists zig-zag trails $\hat{N}_i$ from $u$ to $v_i$ contained in $F$. Since $F$ does not have crowns, these trails do not contain crowns and must be unique. Thus, $\hat{N}_i$ are maximal $N$-fences in $F$. Now let us prove that $G = \bigcup_{i=1}^k \hat{N}_i$ covers $F$. We have to prove that $A(G)=A(F)$. Clearly, we have $A(G) \subseteq A(F)$, so we need to show the other inclusion. Let $a\in A(F)$. By construction of $F$, there exists a maximal zig-zag trail $A$ such that $a$ is an arc of $A$, so let us denote $A = (a_1,\ldots,a_{i-1},a,a_{i+1},\ldots,a_m)$. If $\tail(a_1) = u$ and $\head(a_m) = v_j$ or $\head(a_1) = v_j$ and $\tail(a_m) = u$, for some $j\in [k]$, then $A$ is either $\hat{N}_j$ or $\hat{N}_j$ inverted. In both cases we have $a\in A(\hat{N}_j)$. Else, we have $\head(a_1) = v_j$ and $\head(a_m) = v_l$, for some $j,l\in [k]$. In this case, it is easy to see that $a$ must be either an arc in $\hat{N}_j$ or in $\hat{N}_l$. Therefore, $a\in A(G)$.
\end{pf}

\begin{cor}\label{cor:coverM}
    Let $F$ be a generalized $M$-fence with $End(F) = \{v_1,\ldots,v_k\}$, where $v_i$ are tree-vertices. Let $G = \bigcup_{i=2}^k \hat{M}_i$, where $\hat{M}_i$ are maximal $N$-fences in $F$ going from $v_1$ to $v_i$. Then, $G$ is a covering of $F$. 
\end{cor}

\begin{pf}
    Analogous proof to the one presented in Corollary \ref{cor:coverN}.
\end{pf}

\begin{cor}\label{cor:coverC}
    Let $F$ be a generalized crown with $End(F) = \{v_1,\ldots,v_k\}$, where $v_i$ are tree-vertices. Let $G = C\cup \bigcup_{i=2}^k \hat{N}_c^i$, where $C$ is the only crown of $F$ and $\hat{N}_c^i$ are maximal lower $N$-crowns in $F$ starting at $v_i$. Then, $G$ is a covering of $F$. 
\end{cor}

\begin{pf}
    Analogous proof to the one presented in Corollary \ref{cor:coverN}.
\end{pf}


Now, we have concrete coverings for each extended zig-zag trail that appears in semi-binary tree-based networks. For each extended zig-zag trail $F$, and $G$ constructed in Corollaries \ref{cor:coverN}, \ref{cor:coverM} and \ref{cor:coverC}, depending on the nature of $F$, we will refer to $G$ as an \emph{efficient covering} of $F$.
We will use efficient coverings to construct support trees.

For a covering $G = \cup_{i=1}^k A_i$, with $A_i = (a_1^i,a_2^i,\ldots,a_{m_i}^i)$, we will denote by $d_{2j-1}^i$ the child of $\head(a_{2j-1}^i)$ for $j\in \left[\frac{k_i-1}{2}\right]$. 
The multiplicity of each arc $c_{2j-1}^i := \head(a_{2j-1}^i)d_{2j-1}^i$, in the bulged version, is equal to the indegree of~$\head(a_{2j-1}^i)$ minus one.
In all cherry covers, each of these parallel arcs will be covered in a reticulated cherry shape as the lowest arc. Furthermore, no two parallel arcs appear in the same reticulated cherry shape.
The symmetry of each parallel arc means that the order in which they are covered is irrelevant in a cherry cover. 
Thus, we allow for the arc $c_j^i$ to appear in multiple reticulated cherry shapes.
It is to be understood, that each appearance of such arcs $c_j^i$ refers to different versions of the same parallel arc.


Consider an efficient covering $G = \bigcup_{i=1}^k A_i$ of a generalized $M$ or $N$ fence~$F$, where each of the~$A_i = (a_1^i,a_2^i,\ldots,a_{m_i}^i)$ start at the same node, so that~$a_1^1 = a_1^2 = \ldots = a_1^k$.
Since every pair of maximal fences~$A_i$ has different endpoints, no maximal fence is properly contained in another maximal fence.
This means that every maximal fence~$A_i$ `splits' off from the other fences at some point. 
In particular, each fence~$A_i$ splits at some point from the previous fences $A_j$ with $j<i$. Since $F$ does not have any crowns, once the fence splits, it ends in a branch structure that does not intersect the previous trails again. 
We define the \emph{split number} of a zig-zag trail~$A_i$
as $s_G^{A_i} = \min\{j\in \{1,\ldots,m_i\}\ |\ a_j^i\notin A(\bigcup_{l=1}^{i-1}A_l)\}$. 
 
\begin{lem}\label{lem:genN}
    Let $N$ be a semi-binary tree-based network. Let $F$ be a generalized $N$-fence of $N$ and $G = \bigcup_{i=1}^k \hat{N}_i$ its efficient covering, with $\hat{N}_i=(a_1^i,a_2^i,\ldots,a_{k_i}^i)$. Then, for $i\in [k]$ and $j\in \left[\frac{k_i-1}{2}\right]$, every cherry cover of $N$ contains the reticulated cherry shapes $\{c_{2j-1}^{i},a_{2j}^i,a_{2j+1}^i\}$.
\end{lem}

\begin{pf}
The arc~$c_1^i$ must be covered as the bottom arc of a reticulated cherry shape.
    Since~$a_1^i$ is incident to two reticulations, it cannot be covered as a middle arc.
    This means that all other incoming arcs of~$\tail(c_1^i)$ must be covered as a middle arc, together with~$c_1^i$, i.e., 
    the shape must be of the form~$\{c_1^i, a_2^i, a_3^i\}$. 
    Since the reticulated cherry shapes contain~$a_3^i$, this means in particular that~$\tail(c_3^i)$ has one incoming arc $a_3^i$ that is already covered.
    Then, all other incoming arcs of~$\tail(c_3^i)$ must be covered as a middle arc.
    That is, we must cover the parallel arcs~$c_3^i$ with reticulated cherry shapes~$\{c_3^i, a_4^i, a_5^i\}$.
    This argument can be repeated for all arcs~$c_{2j-1}^i$ for $i\in[k]$ and $j\in [\frac{k_i-1}{2}]$ to obtain a covering of~$F$ and all outgoing arcs, and we are done.
    
\end{pf}

We define the \emph{length} of a generalized $M$-\emph{fence} as the number of arcs in $F$, denoted by $k_F$. For an efficient covering $G = \cup_{i=1}^k \hat{M}_i$ of $F$, let $k_i$ denote the length of each $\hat{M}_i$ and $s_G^{\hat{M}_i}$ its split number with respect to $G$.
It is clear that $k_F = k_1 + \sum_{i=2}^k (k_i-(s_G^{\hat{M}_i}-1))$.

\begin{lem}\label{lem:genM}
    Let~$N$ be a semi-binary tree-based network, and let~$F$ be a generalized $M$-fence of length~$k_F\ge 4$. Let $G = \bigcup_{i=1}^k \hat{M}_i$ be an efficient covering of $F$ for some designated vertex~$v$, with $\hat{M}_i = (a_1^i,\ldots,a_{k_i}^i)$ such that~$\head(a_1^i) = v$. There exists exactly $k_F/2$ distinct ways of covering the arcs of~$F$ and the arcs~$c^i_{2j}$ for~$i\in \left[k\right]$ and $j\in\left[\frac{k_i-2}{2}\right]$ using cherry and reticulated cherry shapes.
\end{lem}

\begin{pf}

    The generalized $M$-fence $F$ contains $k_F/2-k$ reticulations and $k_F-(k+1)$ reticulation arcs, so we must have exactly $k_F/2-1$ reticulated cherry shapes covering the arcs of~$F$ and the outgoing arcs of the reticulations. Thus, $k_F-2$ arcs of $F$ are covered by reticulated cherry shapes.
    The remaining two arcs of~$F$ must be arcs~$a^i_m$ and~$a^j_n$ for some $i,j\in [k]$ and $m\in[k_i]$ and $n\in[k_j]$. Without loss of generality assume that $n>m$.
    By a similar argument to the one given in Lemma \ref{lem:M-fenceSizeCC}, we can see that $i = j$ and $n = m+1$.

    Therefore, every generalized $M$-fence of length~$k_F\ge4$ must be covered by a single cherry shape and~$k_F/2-1$ reticulated cherry shapes in every cherry cover.
    In particular, note that the choice of arcs for the cherry shape determines the rest of the reticulated cherry shapes that cover the rest of the arcs of $F$. 
    Since there are $k_F/2$ potential locations for cherry shapes,
    we conclude that there are $k_F/2$ distinct ways of covering the arcs of~$F$ and the arcs~$\head(a^i_{2j})c^i_{2j}$ for~$i\in [k]$ and $j\in\left[\frac{k_i-2}{2}\right]$ using cherry and reticulated cherry shapes.
\end{pf}

\begin{lem}\label{lem:gencrown}
        Let~$N$ be a semi-binary tree-based network, and let~$F$ be a generalized crown. Let $G = \bigcup_{i=1}^k \hat{N_c}^i\cup C$ be an efficient covering of $F$, where~$C = (b_1,\ldots,b_{2m_0})$ is a crown and $\hat{N_c}^i = (a_1^i,\ldots,a_{m_i}^i, b_{\sigma_i(1)}, \ldots, b_{\sigma_i(2m_0)})$ are the lower~$N$-crowns, for some cyclic permutations $\sigma_i:[2m_0] \rightarrow [2m_0]$. 
    Also, let $c_\ell$ denote the arc between $\head(b_\ell)$ and its child, and $c_j^i$ the arc between $\head(a_j^i)$ and its child. Then, there exist exactly 2 distinct ways of covering the arcs of~$F$, the arcs~$c_{\ell}$ for $\ell\in [2m_0]$ and the arcs~$c^i_{2j}$ for~$j\in\left[\frac{m_i}{2}\right]$ and $i\in[k]$ using reticulated cherry shapes.
\end{lem}

\begin{pf}
    By the semi-binary property and the fact that for each reticulation, only one of its incoming arcs is not covered as the middle arc of a reticulated cherry shape, we have that the arcs of $C$ and the arcs~$c_1,\ldots, c_{2m_0}$ can be covered in two different ways in the bulged version of $N$, as we had in the binary case.

    Once the arcs of the crown are covered, the only arcs left to cover are the arcs $a_1^i,\ldots, a_{m_i}^i$ and the arcs~$c^i_{2j}$ for $j\in\left[\frac{m_i}{2}\right]$ and $i\in [k]$. But since $\hat{N_c}^i$ are lower $N$-crowns, we have $\head(a_{m_i}^i) = \head(b_{\sigma_i(1)}^i) = \head(b_{\sigma_i(2m_0)}^i)$, and we have that either $b_{\sigma_i(1)}$ or $b_{\sigma_i(2m_0)}$ are covered in a reticulated cherry shape containing $c_{\sigma_i(1)}$. Thus, the cherry cover contains the reticulated cherry shape $\{c_{\sigma_i(1)},a_{m_i}^i,a_{m_i-1}^i\}$. 
    Now we have $a_{m_i-1}^i$ covered in a shape that does not contain $c_{m_i-1}^i$. 
    Because of this, every reticulation arc incident to $\tail(c_{m_i-1}^i)$ that is not~$a_{m_i-1}^i$ must be covered as a reticulated cherry shape together with $c_{m_i-1}^i$.
    Therefore, we have a reticulated cherry shape of the form $\{c_{m_i-1}^{i},a_{m_i-2}^i,a_{m_i-3}^i\}$. Note that $m_i$ is even, since $\hat{N}_c^i$ is a lower $N$-crown. Repeating the argument, we have reticulated cherry shapes of the form $\{c_{m_i-2j+1}^{i},a_{m_i-2j}^i,a_{m_i-2j-1}^i\}$ for $1\leq j\leq \frac{m_i}{2}$, and this hold for every $i\in [k]$. Observe that with this argument we covered all the arcs of $F$, as well as the arcs~$c_{\ell}$ for $\ell\in [2m_0]$ and the arcs~$c^i_{2j}$ for~$j\in\left[\frac{m_i}{2}\right]$ and $i\in[k]$.

    Summing up, we have $2$ possibilities of covering the arcs of $C$, and once they are covered, the rest of the shapes that cover the arcs of $F$ are fixed. 
    It follows that every cherry cover of~$N$ covers the arcs of $F$ in one of two possible ways.
\end{pf}

\begin{thm}\label{thm:SBcount}
        Let~$N$ be a semi-binary tree-based network. 
    Let~$A_1,A_2,\ldots, A_m$ denote the generalized $M$-fences of~$N$ of length $k_i\ge 4$ where $i\in [m]$ and let~$B_1,B_2,\ldots,B_c$ denote the generalized crowns of~$N$.
    Then~$N$ has $2^{c-m} k_1 k_2 \cdots k_m$ distinct cherry covers and support trees.
\end{thm}

\begin{pf}
    Let $P$ be a cherry cover. Let $D$ be the extended fence decomposition of $N$. By Corollary \ref{cor:struct}, we know that the only extended zig-zag trails in $D$ are generalized $N$ and $M$-fences and generalized crowns. This decomposition gives us a partition of the arcs of $N$ into disjoint blocks, each one of them of three different types:

    \begin{itemize}
            \item[1)] Generalized $N$-fences without the only arc connecting two reticulations, and adding the outgoing arcs of the rest of each reticulation.
        \item[2)] Generalized $M$-fences together with the outgoing arcs of each reticulation.
        \item[3)] Generalized crowns without the arcs connecting reticulations and adding the outgoing arcs of the reticulations.
    \end{itemize}

    By Lemma \ref{lem:genN}, we know that the type 1) structures can be covered in $1$ way. By Lemma \ref{lem:genM}, for a generalized $M$-fence of length $k$, we have that type 2) structures can be covered in $k/2$ ways. And, using Lemma \ref{lem:gencrown}, type 3) structures can be covered in $2$ different ways. It follows that there are

    $$2^c\times \prod_{i=1}^m\frac{k_i}{2}$$

    many possible cherry covers.
\end{pf}

\begin{thm}\label{thm:ForbiddenStructure}
    Let $N$ be a semi-binary network, and let $D$ be its extended fence decomposition. Then, $N$ is tree-based if and only if, for all zig-zag fences $F\in D$ in $N$, the following holds:
    \begin{itemize}
        \item[(i)] $F$ does not contain any $W$-fence.
        \item[(ii)] $F$ does not contain any lower $W$-crown.
        \item[(iii)] $F$ contains at most $1$ binary crown.
    \end{itemize}
\end{thm}

\begin{pf}
    The only if part follows directly from Lemma \ref{lem:SBTB-forbidden}. For only if part, let us assume that $N$ does not have any maximal extended zig-zag trail $F$ satisfying \textit{(i)}, \textit{(ii)} or \textit{(iii)}. Thus, each of the maximal extended zig-zag trails are generalized $N$-fences, generalized $M$-fences or generalized crowns. Let us construct a cherry cover of $N$. We can partition the arcs of $N$ in its set of maximal extended zig-zag trails, which gives another partition in blocks of the three different types:

    \begin{itemize}
        \item[1)] Generalized $N$-fences without the only arc connecting two reticulations, and adding the outgoing arcs of the rest of each reticulation.
        \item[2)] Generalized $M$-fences together with the outgoing arcs of each reticulation.
        \item[3)] Generalized crowns without the arcs connecting reticulations and adding the outgoing arcs of the reticulations.
    \end{itemize}

    If we can cover each block with cherry shapes and reticulated cherry shapes, that would give us a cherry cover for the whole network.

    Therefore, let us consider, for each block of type 1), the set of shapes and reticulated cherry shapes defined in Lemma \ref{lem:genN}, and for the structures of type 2) and 3), one of the sets defined in the proof of Lemmas \ref{lem:genM} and \ref{lem:gencrown}, respectively. Clearly, this gives us a proper covering of the blocks in cherry shapes and reticulated cherry shapes, so the union of all such sets gives us a covering for the whole network in disjoint shapes. This is a cherry cover for $N$, and by Theorem \ref{thm:TBiffCC}, $N$ is tree-based. 
    \end{pf}

\section{Characterization of semi-binary tree-child networks}
\label{sec:TC}

Recall that $\mathcal{P}_{N}$ denotes the set of all cherry covers of a network~$N$. In this section, we characterize semi-binary tree-child networks \cite{cardona2008comparison} in terms of $|\mathcal{P}_{N}|$.

\begin{lem}\label{lem:cota2}
Let $N$ be a semi-binary network with $k$ reticulations with indegrees $r_1,\ldots,r_k$, respectively. Then,  $|\mathcal{P}_{N}|\leq \prod_{i=1}^k r_i$.
\end{lem}

\begin{pf}
    Observe that, in a cherry cover, once we select the reticulated cherry shapes, the rest of the cherry cover is fixed. Moreover, in the semi-binary case, reticulated cherry shapes are determined by its middle arc. So we need to count the possible selections for this middle arc. For a reticulation vertex $u$ with $\indeg u=d$, there must be $d-1$ reticulation arcs incident to $u$ covered as the middle arc of a reticulated cherry shape, and only $1$ arc without this property. Now, if we denote by $\{u_i\}_{i=1}^k$ and $\{r_i\}_{i=1}^k$, the set of reticulations and its indegrees, respectively, we have $r_i$ possibilities to select the only arc, incident to $u_i$, that would not be covered as a middle arc of a reticulated cherry shape in a cherry cover. 
    Thus, we have $\prod_{i=1}^k r_i$ possible ways to select such arcs, and therefore at most this many cherry covers. 
    Note that for some combination of reticulation arc selections, it is possible for a cherry cover to not exist.
\end{pf}

\begin{prop}\label{prop:TC_CV}
Let $N$ be a semi-binary network with $r$ reticulations. Let $r_1,\ldots,r_k$ denote the indegrees of each reticulation. Then, $N$ is tree-child if, and only if,  $|\mathcal{P}_N| = \prod_{i=1}^k r_i$.
\end{prop}

\begin{pf}
    First, assume that $N$ is tree-child. Observe that in this case the extended fence decomposition of $N$ only contains $M$-fences of size $2$ and $k$ generalized $M$-fences of size $2r_i$. By Theorem \ref{thm:SBcount}, the result follows. 
    
    Now, assume that $|\mathcal{P}_N| = \prod_{i=1}^k r_i$ and suppose that $N$ is not a tree-child network. Then, there exists a non-leaf vertex $u$ such that all its children are reticulations. By the semi-binary property, $u$ either has $1$ child (if it is a reticulation) or $2$ children (if it is a tree node). Assume that $u$ is a reticulation, and let $v$ be its only reticulation child and $h$ the only child of $v$. 
    
    Consider the graph $N'$ obtained by: subdividing the arc $uv$ with a vertex $u'$; adding a new vertex $v'$; and adding the arc $u'v'$. It is clear that the result of this operation gives a semi-binary network. Observe that, by the construction of $N'$, any cherry cover of $N$ determines, over the bulged version of $N'$, a unique covering of all the arcs but $u'v'$ and $u'v$. Therefore, for any cherry cover $P\in |\mathcal{P}_N|$, if we add the cherry shape $ \{u'v',u'v\}$, we get a cherry cover for the bulged version of $N'$, so we have $P\cup\{u'v',u'v\}\in \mathcal{P}_{N'}$. Thus, $|\mathcal{P}_{N'}| \geq |\mathcal{P}_{N}| = \prod_{i=1}^k r_i$.
    
    Also, observe that over the bulged version of $N$, we have many parallel arcs between $v$ and $h$, let $a$ be one of them. Let $P$ be a cherry cover of $N$. Then, $P$ must cover the arc $a$ in some reticulated cherry shape, let us denote by $\{a,wv,wl\}$, where $w$ is a parent of $v$ such that it is tree node and $l$ is the child of $w$ different from $v$. Observe that $w\neq u$ since $u$ is a reticulation. Consider the set $P' = (P\setminus \{a,wv,wl\})\cup\{a,u'v,u'v'\}\cup \{wv,wl\}$. It is clear that $P'\setminus \{a,u'v,u'v'\}$ is not a cherry cover of $N$, since it does not cover the arc $a$, and $P'$ does not contain intersecting shapes. Therefore, $P'$ represents a different cherry cover than the ones that we could build from $\mathcal{P}_{N}$. 
    Noting that the indegrees of the reticulations have not changed, this means that $|\mathcal{P}_{N'}| \geq \prod_{i=1}^k r_i+1$, contradicting Lemma \ref{lem:cota2}.

    If $u$ is a tree node and $v$ and $w$ represent its children, let us construct $N''$ from $N$ by subdividing the arc $uv$ with a node $u'$, adding a new vertex $v'$, and adding an arc $u'v'$. A similar argument as in the previous case leads us to a contradiction.
\end{pf}

\section{Conclusion}\label{sec:Conc}

To summarize our contributions, we have shown that the number of cherry covers coincide with the number of support trees for a semi-binary network (\Cref{prop:CC=ST-SB}), and we showed that for binary tree-based networks, there is a natural bridge between its fence decomposition and its cherry covers (\Cref{thm:NumberOfCherryCovers}).
We have extended the notion of a fence decomposition to the non-binary setting, by introducing new types of fences in \Cref{sec:extension}.
Using this new characterization, every non-binary network can be decomposed into a unique extended fence decomposition (\Cref{prop:extended_decomp}).
With this well-defined decomposition, we proved the generalized version of \Cref{thm:NumberOfCherryCovers}, by showing that one can count the number of cherry covers (and therefore the number of support trees) of a semi-binary tree-based network (\Cref{thm:SBcount}).
We also gave a new characterization of semi-binary tree-based networks based on forbidden structures within the extended fence decomposition (\Cref{thm:ForbiddenStructure}).
Finally, we gave a new characterization of semi-binary tree-child networks based on the number of cherry covers (\Cref{prop:TC_CV}).

We now discuss potential future directions regarding cherry covers and extended fence decompositions, starting with open problems on generalizations to the non-binary setting.
For example, can we show that for non-binary tree-based networks, the number of cherry covers coincide with the number of support trees?
At first glance, this does not seem true. 
To give some justification, recall that within a cherry cover for a non-binary network, each arc may be covered more than once (for the formal definition, see Definition 2.10 of~\cite{van2021unifying}).
Consider a non-binary tree~$T$ on three leaves~$\{x,y,z\}$ with a common parent $t$. 
Including the root~$\rho$, this is a tree on 5 vertices~$\{\rho, t, x, y, z\}$ and 4 arcs $\{\rho t, tx, ty, tz\}$.
By definition, $\{\{tx,ty\},\{ty,tz\}\}$ and $\{\{tx,ty\},\{tx,tz\}\}$ are both cherry covers of~$T$. 
Since~$T$ is a tree, it has only one support tree, namely the tree itself.
Thus, naively, the equality $|\mathcal{S}_T|= |\mathcal{P}_T|$ does not hold. 
However, we note that there is exactly one \emph{maximal cherry cover} $\{\{tx,ty\},\{tx,tz\},\{ty,tz\}\}$ of~$T$, which has maximal size over all possible cherry covers. 
Therefore, we conjecture the following.

\begin{conj}\label{conj:S(N)=MaxP(N)}
    Let~$N$ be a non-binary network. Let~$\mathcal{S}_N$ and~$\mathcal{P}_N$ denote the set of support trees and maximal cherry covers for~$N$, respectively. Then~$|\mathcal{S}_T|= |\mathcal{P}_T|$.
\end{conj}

In another direction, one can attempt for a forbidden structure characterization for non-binary tree-based networks.
Here, the analysis of the extended fence decompositions quickly becomes more complex, as upper $W$-crowns and upper $N$-crowns are not immediately forbidden due to the relaxation of degree constraints.


\newpage
\noindent \sloppy \textbf{Acknowledgements.} JCP and PVL was partially supported by the Spanish Ministry of Economy and Competitiveness and European Regional Development Fund project PID2021-126114NB-C44 funded by MCIN/AEI/10.13039/501100011033 and by “ERDF A way of making Europe.”

\bibliographystyle{alpha}
\bibliography{mybibfile}

\end{document}